\documentclass[lnbip]{svmultln}
\usepackage{graphicx}
\usepackage{amsmath}
\usepackage{bm}        
\usepackage{verbatim}   
\begin{document}

\title{Pure state `really' informationally complete with rank-1 POVM}

\author{Yu Wang\inst{1,2} , Yun Shang\inst{1,3}}
\authorrunning{ }
\titlerunning{ }
\institute{Institute of Mathematics, AMSS, CAS, Beijing, 100190, China,
\and
University of Chinese Academy of Sciences, Beijing, 100049, China\\
\and
NCMIS, AMSS, CAS, Beijing, 100190, China.\\
\email{shangyun602@163.com}}
\maketitle

\begin{abstract}
 What is the minimal number of elements in a  rank-1 positive-operator-valued measure (POVM)  which can uniquely determine any pure state in $d$-dimensional Hilbert space $\mathcal{H}_d$?
The known result is that the number is no less than $3d-2$.
We show that this lower bound is not tight except for $d=2$ or 4.
Then we give an upper bound of $4d-3$.
For $d=2$, many rank-1 POVMs with four elements can determine any pure states in $\mathcal{H}_2$.
For $d=3$, we show eight is the minimal number by construction.
For $d=4$, the minimal number is in the set of $\{10,11,12,13\}$.
We show that if this number is greater than 10, an unsettled open problem can be solved that three orthonormal bases can not distinguish all pure states in $\mathcal{H}_4$.
For any dimension $d$, we construct $d+2k-2$ adaptive rank-1 positive operators for the reconstruction of any unknown pure state in $\mathcal{H}_d$, where $1\le k \le d$.

\end{abstract}

\section{Introduction}
One of the central problems in quantum science and technology is the estimation of an unknown quantum state, via the measurements on a large number of copies of this state.
Quantum state tomography is the process of determining an arbitrary unknown quantum state with appropriate measurement strategies.

A quantum state $\rho$ in $d$-dimensional Hilbert space $\mathcal{H}_d$ is described by a density matrix, namely by a positive semi-definite, unit-trace $d\times d$ matrix as $S_d$.
A generalized measurement can be described by a positive operator-valued measure (POVM) \cite{Nielsen_2000}.
The POVM elements, $E_k$, satisfy the completeness condition: $\sum_kE_k=I$.
Performing this measurement on a system in state $\rho$, the probability of the $k$-th outcome is given from the Born rule, $p_k=\mbox{tr}(\rho E_k)$.
If the statistics of the outcome probabilities are sufficient to uniquely determine the state, the POVM is regarded as \emph{informationally complete} (IC) \cite{Prugovcki_1991}.

The IC-POVM can give a unique identification of an unknown state, which should distinguish any pair of different states from the statistics of probabilities.
For example, we consider a POVM, $\{\frac{1}{4}(|0\rangle\pm|1\rangle)(\langle0|\pm\langle1|), \frac{1}{4}(|0\rangle\pm i|1\rangle)(\langle0|\mp i\langle1|)\}$. It is not an IC-POVM, as the statistics of the outcome probabilities for states $\{|0\rangle,|1\rangle\}$ under this measurement are the same.
For any different quantum states $\rho_1,\rho_2\in S_d$, an IC-POVM should distinguish them from the statistics of the outcome probabilities.
That is to say,
we have $\mbox{tr}(\rho_1 E_k)\ne \mbox{tr}(\rho_2 E_k)$ for some elements $E_k$.

We know that a quantum state $\rho$ in $\mathcal{H}_d$ is specified by $d^2-1$ real parameters.
The number is reduced by one because $\mbox{tr}(\rho)=1$.
Caves \emph{et al.} constructed an IC-POVM which contains the minimal $d^2$ rank-1 elements \cite{Caves_2002}, i.e., multiples of projectors onto pure states.
If $d+1$ mutually unbiased bases (MUBs) exist in $\mathcal{H}_d$, we can construct an IC-POVM with $d(d+1)$ elements \cite{Wootters_1989}.
MUBs have the property that all inner products between projectors of different bases labeled by $i$ and $j$ are equal to $1/d$.
Another related topic is the symmetric informationally complete positive operator-valued measure (SIC-POVM) \cite{Renes_2004}.
It is comprised of $d^2$ rank-1 operators.
The inner products of all different operators are equal.
This SIC-POVM appears to exist in many dimensions.

For a state $\rho$ in $n$-qubit system, $d=2^n$.
Thus the cost of measurement resource with these measurement strategies grows exponentially with the increase of number $n$.
It is important to design schemes with lower outcomes to uniquely determine the state.
This is possible when we consider a priori information about the states to be characterized.

Denote the rank of a density matrix for state $\rho$ as $k$, $1\le k \le d$.
And make a decomposition that $S_d=\oplus_{k=1}^d S_{d,k}$, where $S_{d,k}$ is the set of all the density matrices with rank $k$.
When $k=1$, the state in $S_{d,1}$ is pure.
A pure state is specified by $d$ complex numbers, which correspond to $2d$ real numbers.
For the reason of normalization condition and freedom of a global phase, there are $2d-2$ independent real numbers totally.

Flammia, Silberfarb, and Caves \cite{Flammia_2005} showed that any POVM with less than $2d$ elements can not distinguish all pair of different states $\rho_1,\rho_2$ in $ S_{d,1}$, not even in a subset $\tilde{S}_{d,1}$, where $S_{d,1}\setminus\tilde{S}_{d,1}$ is a set of measure zero.
They gave a definition of pure-state informationally complete (PSI-complete) POVM, whose outcome probabilities are sufficient to determine any pure
states (up to a global phase), except for a set of pure states that is dense only on a set of measure zero.
That is to say, if a pure state was selected at random, then with probability 1 it would be located in $\tilde{S}_{d,1}$ and be uniquely identified.
A PSI-complete POVM with $2d$ elements is constructed, but not all the elements in this POVM are rank-1.
They constructed another PSI-complete POVM with $3d-2$ rank-1 elements and conjectured that there exists a rank-1 PSI-complete POVM with $2d$ elements.

Finkelstein proved this by a precise construction \cite{Finkelstein_2004}.
Moreover, he gave a strengthened definition of PSIR-completeness, which indicates that all pure states are uniquely determined.
For any pair of different pure states $\rho_1,\rho_2\in S_{d,1}$, a PSIR-complete POVM should distinguish them.
He showed that a rank-1 PSIR-complete POVM must have at least $3d-2$ elements and wondered whether we could reach the lower bound of $3d-2$.

There are a series of studies on the relevant topic.
For any pair of different states $\rho_1,\rho_2\in S_{d,1}$, Heinosaari, Mazzarella, and Wolf gave the minimal number of POVM elements to identify them \cite{Heinosaari_2013}.
The number is $4d-3-c(d)\alpha(d)$, where $c(d)\in[1,2]$ and $\alpha(d)$ is the number of ones appearing in the binary expansion of $d-1$;
the results in papers \cite{Mondragon_2013,Jaming_2014,Carmeli_2015} showed that four orthonormal bases, corresponding to four projective measurements, can distinguish all pure states.
For any pair of different states $\rho_1\in S_{d,1}$, $\rho_2\in S_d$,
Chen \emph{et al.} showed that a POVM must contain at least $5d-7$ elements to distinguish them \cite{Chen_2013};
Carmeli \emph{et al.} gave five orthonormal bases that are enough to distinguish them\cite{Carmeli_2016}.
For a state in $S_{d,k}$, it can be reconstructed with a high probability with $rd\log^2(d)$ outcomes via compressed sensing techniques \cite{Gross_2010}.
Goyeneche \emph{et al.} \cite{Goyeneche_2015} constructed five orthonormal bases to determine all the coefficients of any unknown input pure states.
The first basis is fixed and used to determine a subset $s_{d,1}\subset S_{d,1}$, where the pure state belongs to.
The other four bases are used to uniquely determine all the states in $s_{d,1}$.

%



In this paper, we consider the pure-state version of informational completeness with rank-1 POVM.
Firstly, we show that the lower bound of $3d-2$ is not tight in most of the cases.
It can be reached when $d=2$ and possibly be reached when $d=4$.
Then we show a result that there exist a large number of rank-1 PSIR-complete POVMs with $4d-3$ elements.
Secondly, we make a discussion about the rank-1 PSIR-complete POVMs when $d=2,3,4$.
For dimension $d=2$ and $d=3$, we construct the rank-1 PSIR-complete POVMs with the minimal number of elements, which are 4 and 8 correspondingly.
All the coefficients of an unknown pure state in $\mathcal{H}_2$ and $\mathcal{H}_3$ can be calculated by these POVMs.
For dimension $d=4$, the minimal number is in the range of $\{10,11,12,13\}$.
If it is bigger than 10, an answer can be given to a related unsolved problem, i.e., three orthonormal bases can not distinguish all pure states in $\mathcal{H}_4$.
Lastly, we construct $d+2k-2$ rank-1 positive self-adjoint operators for the tomography of any input pure states in $\mathcal{H}_d$, here $1\le k\le d$.
This is an adaptive strategy.
For any input pure state, we use $d$ operators to determine a subset $s_{d,1}\subset S_{d,1}$, where the pure state belongs to.
Together with the other $2k-2$ operators, we can uniquely determine all the pure states in $s_{d,1}$.
Thus using this adaptive method, any input pure states can be determined with at most $3d-2$ rank-1 operators.

\section{The upper and lower bounds}

In this section, we will give the upper and lower bounds of the minimal number of elements in a rank-1 PSIR-complete POVM.
Denote this minimal number as $\textbf{m}_1(d)$.
It is in the range of $[4d-3-c(d)\alpha(d),4d-3]$.

\subsection{Feasibility of 3d-2 for PSIR-complete}

In this part, we show that a rank-1 PSIR-complete POVM with $3d-2$ elements possibly exists when dimension $d=2$ or 4.
For the other dimensions, any rank-1 POVM with $3d-2$ elements cannot be PSIR-complete.
Firstly, we introduce the concept of PSIR-complete.

\emph{Definition 1}: (PSI really-completeness \cite{Finkelstein_2004}). A pure-state informationally really complete POVM on a $d$-dimensional quantum system $\mathcal{H}_d$ is a POVM whose outcome probabilities are sufficient to uniquely determine any pure state (up to a global phase).

As we introduced above, the PSIR-complete POVM can distinguish any pair of different states $\rho_1,\rho_2\in S_{d,1}$.
Neglecting the restriction of rank-1, we denote $\textbf{m}_0(d)$ to be the minimal number of elements in a PSIR-complete POVM.
Certainly, a rank-1 PSIR-complete POVM is PSIR-complete.
Thus $\textbf{m}_1(d)\ge \textbf{m}_0(d)$.
From the result in \cite{Heinosaari_2013}, $\textbf{m}_0(d)=4d-3-c(d)\alpha(d)$, where $c(d)\in[1,2]$ and $\alpha(d)$ is the number of ones appearing in the binary expansion of $d-1$.
From the conclusion by Finkelstein, $\textbf{m}_1(d)\ge 3d-2$.
But it is not clear when they are equal or whether a greater number than $3d-2$ might be required.
Now we compare the size of $\textbf{m}_0(d)$ and $3d-2$.

Let $f(d)=4d-3-c(d)\alpha(d)-(3d-2)$.
By the definition of $\alpha(d)$, we have $\log d\ge \alpha(d)$.
So $f(d)>d-1-2\log d$.
Define $g(d)\equiv d-1-2\log d$.
Then $g^{\prime}(d)=1-2/d$.
If $d>2$, it holds that $g^{\prime}(d)>0$.
And when $d=8$, $g(8)=1>0$.
So when $d\in[8,+\infty)$, $\textbf{m}_0(d)>3d-2$.
When $d\in[2,7]$, the true value of $\textbf{m}_0(d)$ is given in \cite{Heinosaari_2013}.
We have $\textbf{m}_0(2,3,4,5,6,7)=(4,8,10,16,18,23)$.
We compare this with the value of $3d-2$: $(4,7,10,13,16,19)$.
As a result, only when $d=2$ or 4, $\textbf{m}_0(d)$ can be $3d-2$.
For the other dimensions, $\textbf{m}_1(d)\ge\textbf{m}_0(d)>3d-2$.


\subsection{The upper bound of $4d-3$}

In this section, we show that $4d-3$ is the upper bound of $\textbf{m}_1(d)$.
This upper bound is given by constructing rank-1 POVMs from the minimal sets of orthonormal bases which can determine all pure states in $\mathcal{H}_d$.

\emph{Definition 2}: Let $\mathcal{B}_{0}=\{|\phi_0^k\rangle\}$,$\cdots$,$\mathcal{B}_{m-1}=\{|\phi_{m-1}^{k}\rangle\}$ be $m$ orthonormal bases of $\mathcal{H}_d$, $k=0,\cdots,d-1$. For different pure states $\rho_1, \rho_2\in\mathcal{H}_d$, they are distinguishable if
 \begin{equation} \label{inequation}
\mbox{tr}(\rho_1|\phi_{j}^{k}\rangle\langle\phi_{j}^{k}|)\ne\mbox{tr}(\rho_2|\phi_{j}^{k}\rangle\langle\phi_{j}^{k}|)
\end{equation}
for some $|\phi_j^k\rangle$. If any pair of different pure states is distinguishable by $\mathcal{B}_{0},\cdots,\mathcal{B}_{m-1}$, the bases $\{\mathcal{B}_j\}$ can distinguish all pure states \cite{Carmeli_2015}.


Obviously $m$ bases correspond to $m\cdot d$ rank-1 projections.
$E_j^k=|\phi_j^k\rangle\langle\phi_j^k|$, $j=0,\cdots,m-1$, $k=0,\cdots,d-1$.
Since $\sum_{k=0}^{d-1}E_j^k=I$, we have $\mbox{tr}(\rho I)=1$ for all pure state $\rho$.
One projection for each basis can be left out as the probability can be expressed by others.
Thus $m(d-1)$ rank-1 self-adjoint operators can distinguish all pure states.
Can these operators be transformed to a rank-1 PSIR-complete POVM?

From the proposition 3 in paper \cite{Heinosaari_2013}, we know that $m(d-1)$ self adjoint operators can be used to construct a POVM with $m(d-1)+1$ elements.\\
$A_j^k\equiv(\frac{1}{2}I+\frac{1}{2}\|E_j^k\|^{-1}E_j^k)/[m(d-1)]$. $j=0,\cdots,m-1$, $k=0,\cdots,d-2$.
Then $O\leq A_j^k\leq I/m(d-1)$ and by setting the new element $A\equiv I-\sum\nolimits_{j,k}A_j^k$ we get a new POVM.
This POVM have the same power with the self-adjoint operators $\{E_k\}$, as there exists a bijection between the outcome probabilities of both sides.
But not all of the elements are rank-1.
The following conversion can keep the elements of transformed POVM to be rank-1.

\textbf{Rank-1 conversion:} Given $n$ rank-1 positive self-adjoint operators $\{E_k:k=1,\cdots,n\}$, $G=\sum\nolimits_{k=1}^dE_k>0$, a rank-1 POVM denoted by $\{F_k: k=1,\cdots,n\}$ can be constructed.
$F_k=G^{-1/2}E_kG^{-1/2}$ and $\sum_{k=1}^nF_k=I$.

From the discussion in \cite{Flammia_2005,Finkelstein_2004} , if positive operators $\{E_k\}$ are informationally complete with respect to generic pure states (a set of measure zero can be neglected), and they can determine all (normalized and unnormalized) pure states in this set, $\{F_k\}$ is a PSI-complete POVM.
Furthermore, if positive operators $\{E_k\}$ are informationally complete with respect to all pure states, can the converted POVM $\{F_k\}$ be PSIR-complete? Here we give a sufficient condition.

\begin{theorem} Let $\{E_k\}$ be a set of rank-1 positive self-adjoint operators, whose outcome probabilities are sufficient to uniquely determine all pure states (up to a global phase).
Some of the elements satisfy the following condition:
\begin{equation}
\sum\nolimits_{k\in B} E_k= I.
\label{condition}
\end{equation}
After the rank-1 conversion, we will get a rank-1 PSIR-complete POVM $\{F_k\}$.
\end{theorem}
\begin{proof}
Here we prove that any pair of different pure states is distinguishable by this POVM.

Let $\rho_1$ and $\rho_2$ be an arbitrary pair of different pure states.
Define $q_i=\mbox{tr}(G^{-1}\rho_i)$ for $i=1,2$.
As $G=I+\sum_{k\notin B}E_k$, we have $\mbox{det}(G)\ne0$.
So $G$, $G^{-1/2}$ and $G^{-1}$ are of full rank. And $q_i=\mbox{tr}(G^{-1}\rho_i)\ne0$.

Then define another pair of pure states $\sigma_i=G^{-1/2}\rho_iG^{-1/2}/q_i$, $i=1,2$.
For any $k$,
\begin{equation}
\mbox{tr}(F_k\rho_i)=q_i\mbox{tr}(E_k\sigma_i).
\end{equation}

When pure states $\sigma_1$ and $\sigma_2$ are the same, as $\rho_1\ne\rho_2$,  the number $q_1$ should not be equal to $q_2$.
Thus $\mbox{tr}(F_k\rho_1)\ne\mbox{tr}(F_k\rho_2)$ for any $k$.

When pure states $\sigma_1$ and $\sigma_2$ are different, by the assumption of $\{E_k\}$, there exists some $E_k$ satisfying $\mbox{tr}(E_k\sigma_1)\ne\mbox{tr}(E_k\sigma_2)$.
If $q_1=q_2$, then $\mbox{tr}(F_k\rho_1)\ne\mbox{tr}(F_k\rho_2)$.
If not, we have $\sum_{k\in B}\nolimits\mbox{tr}(E_k\sigma_1)=\sum_{k\in B}\nolimits\mbox{tr}(E_k\sigma_2)=1$.
As we have the assumption $\sum_{k\in B}\nolimits\mbox{tr}(E_k)=I$.
Then $\sum_{k\in B}\nolimits\mbox{tr}(F_k\sigma_1)\ne\sum_{k\in B}\nolimits\mbox{tr}(F_k\sigma_2)$. Thus it can also be deduced that $\mbox{tr}(F_k\rho_1)\ne\mbox{tr}(F_k\rho_2)$ for some $k\in B$.

So the POVM $\{F_k\}$ can distinguish the different pure states $\rho_1,\rho_2\in S_{d,1}$.
This indicates that given a set of outcome probabilities $\{p_k\}$,
there is a unique pure state $\rho$ such that $p_k=\mbox{tr}(\rho F_k)$ for all $k$.
For any other different pure state $\sigma$, we can always get $\mbox{tr}(\sigma F_k)\ne p_k$ for some $k$.
Thus $\{F_k\}$ is enough to uniquely determine any pure states from the other different pure states.
With the prior knowledge that the state is pure, it can be uniquely determined.
\end{proof}

\emph{Remark:}
In this proof, we consider the case where $\{E_k\}$ can distinguish all pure states.
We can make a extension to this theorem.
If $\{E_k\}$ can distinguish all different states $\rho_1,\rho_2\in \mathcal{H}_d$, and the equation \ref{condition} still holds, then the conversed POVM $\{F_k\}$ can is informationally complete with respect to all quantum states, pure or mixed.

\begin{theorem}
Assume that $m$ orthonormal bases can distinguish all pure states in $\mathcal{H}_d$, a large number of PSIR-complete POVMs with $m(d-1)+1$ rank-1 elements can be constructed.
\end{theorem}

\begin{proof}
Denote these orthonormal bases as $\{\mathcal{B}_j\}$, $j=0,\cdots,m-1$.
The elements in basis $\mathcal{B}_j$ are $\{|\phi_j^k\rangle\}$, $k=0,\cdots,d-1$.
Now we pick up $m(d-1)+1$ elements from these bases.
We can randomly choose one basis $\mathcal{B}_j$ and keep all the elements in it.
The corresponding projectors satisfy $\sum_{k=0}^{d-1}|\phi_j^k\rangle\langle\phi_j^k|=I$.
Then we select $d-1$ elements at random from each one of the other bases.
Thus we will get a set of $m(d-1)+1$ elements.
There are $m\cdot d^{m-1}$ collections totally.

Each collection will correspond to $m(d-1)+1$ rank-1 projectors, which can distinguish all pure states.
They satisfy the condition in Theorem 1.
After the rank-1 conversion, we will get a rank-1 PSIR-complete POVM with $m(d-1)+1$ elements.
Moreover, we can construct a large number of PSIR-complete POVMs for each collection.
Denote the projectors to be $\{E_1,\cdots,E_d,E_{d+1},\cdots,E_{m(d-1)+1}\}$, where $\sum_{k=1}^{d}E_k=I$.
We can multiply $E_j$ by an arbitrary non-negative number $e_j$, where $j=d+1,\cdots,m(d-1)+1$.
So a new set of operators is constructed, $\{E_1,\cdots,E_d,e_{d+1}\cdot E_{d+1},\cdots,e_{m(d-1)+1}\cdot E_{m(d-1)+1}\}$.
They also satisfy the condition in Theorem 1.
The proof is complete.
\end{proof}

Various researches focus on the minimal number of orthonormal bases that can distinguish all pure states \cite{Moroz_1983,Moroz_1984,Moroz_1994,Mondragon_2013,Jaming_2014}.
This problem is almost solved.
The minimal number of orthonormal bases is summarized in \cite{Carmeli_2015}.
Moreover, four bases are constructed from a sequence of orthogonal polynomials.
For dimension $d=2$, at least three orthonormal bases are needed to distinguish all pure quantum states.
For $d=3$ and $d\ge 5$, the number is four.
For $d=4$, four bases are enough but it is not clear whether three bases can also distinguish.

So we can give the upper bound of $\textbf{m}_1(d)$.
When $d=2$, $\textbf{m}_1(2)=4$.
When $d\ge 3$, $\textbf{m}_1(d)=4d-3$.

%
%

\section{Rank-1 PSIR-complete POVMs  for $\mathcal{H}_2$, $\mathcal{H}_3$ and $\mathcal{H}_4$}

In this section, we will present some results about the rank-1 PSIR-complete POVMs for lower dimensions $d$.
In Figure \ref{1}, we show the relations between different kinds of informationally complete POVM.
An IC-POVM is a PSIR-complete POVM.

\begin{figure}
  \centering
  \includegraphics[width=4cm]{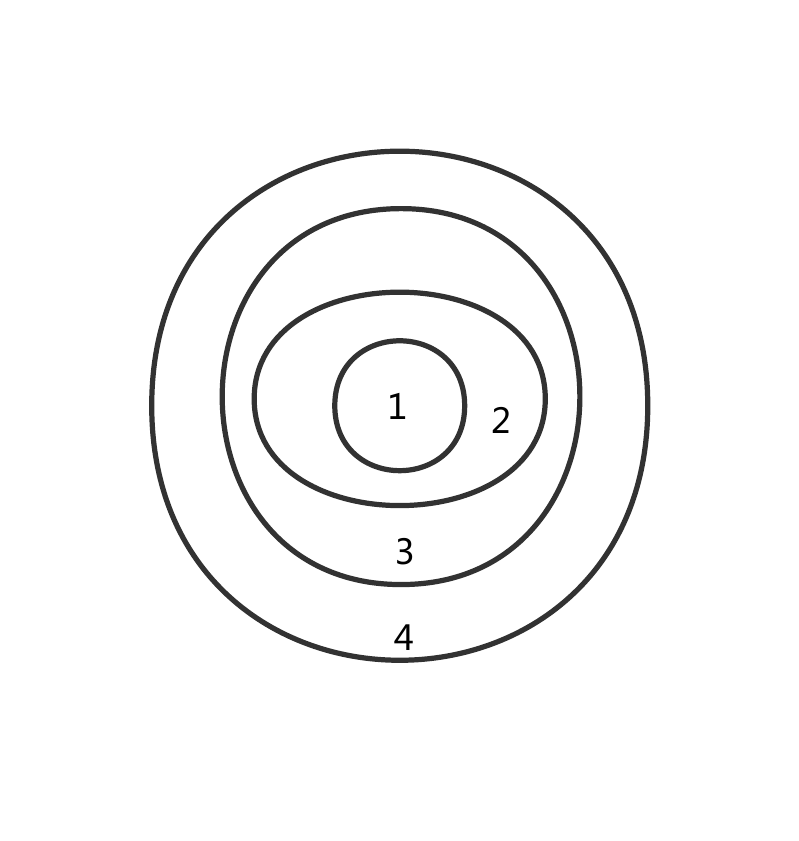}\\
  \caption{The relations of different kinds of informationally complete POVM. The labels $\{1,2,3,4\}$ stand for SIC-POVM, IC-POVM, PSIR-complete POVM, PSI-complete POVM respectively. 
  For example, an IC-POVM is a PSIR-complete POVM.}
  \label{1}
\end{figure}

\subsection{d=2}

For dimension $d=2$, four is the minimal number of elements in a rank-1 PSIR-complete POVM.
One example showed in \cite{Flammia_2005} is the following:
\begin{equation}
 E^c=a^c I+b^c\mathbf{n}^c\cdot \vec{\sigma}, \ c=1,2,3,4.
\end{equation}
The parameter: $a^c=b^c=1/4$.
$\mathbf{n}^1=(0,0,1)$,
$\mathbf{n}^2=(2\sqrt{2}/3,0,-1/3)$,
$\mathbf{n}^3=(-\sqrt{2}/3,\sqrt{2}/3,-1/3)$,
$\mathbf{n}^4=(-\sqrt{2}/3,-\sqrt{2}/3-1/3)$.
$\vec{\sigma}=(\sigma_x,\sigma_y,\sigma_z)$.
This is also a SIC-POVM.
It can distinguish all quantum states in $\mathcal{H}_2$, pure or mixed.
There are two SIC-POVMs for $d=2$ introduced in paper \cite{Renes_2004}.
The other SIC-POVM is used \cite{Rehacek_2004}, which shows the efficiency of qubit tomography.

Now we can construct 12 rank-1 IC-complete POVMs with four elements.
There are three mutually unbiased bases for $d=2$.
\begin{equation}
\mathcal{B}_0=\{|0\rangle,|1\rangle\},
\mathcal{B}_1=\{(|0\rangle\pm|1\rangle)/\sqrt{2}\},
\mathcal{B}_2=\{(|0\rangle\pm i|1\rangle)/\sqrt{2}\}.
\end{equation}
These three mutually unbiased bases can distinguish all quantum states in $\mathcal{H}_2$.
We can select four elements as introduced in Theorem 2.
There are 12 collections totally.
For example, the elements for one collection are $|0\rangle,|1\rangle, (|0\rangle+|1\rangle)/\sqrt{2},(|0\rangle+i|1\rangle)/\sqrt{2}$.
The corresponding rank-1 projectors are $|0\rangle\langle0|$, $|1\rangle\langle1|$, $(|0\rangle+|1\rangle)(\langle0|+\langle1|)/2$ and $(|0\rangle+i|1\rangle)(\langle0|-i\langle1|)/2$.
After the rank-1 conversion, we will get a rank-1 POVM with 4 elements.
Interestingly, this POVM is the special case when $d=2$ constructed by Caves \emph{et al.} \cite{Caves_2002}.

%

\subsection{d=3}
For dimension $d=3$, there are four mutually unbiased bases.
By Theorem 2, we have $4\times3^3$ collections with 9 elements.
We can construct rank-1 IC-complete POVMs with 9 elements from each selection.
By a reference to Heinosaari \emph{et al.} \cite{Heinosaari_2013}, $\textbf{m}_0(3)=8$.
So the minimal number of elements is either 8 or 9 for a rank-1 PSIR-complete POVM.
Now we show that this number is 8 by constructing 8 rank-1 operators satisfying Theorem 1.
After the rank-1 conversion, we will get a PSIR-complete POVM with 8 elements.
%
The operators are as follows:

$E_0=|0\rangle\langle 0|$,
$E_1=|1\rangle\langle 1|$,
$E_2=|2\rangle\langle 2|$,
$E_3=(|0\rangle+|1\rangle)(\langle0|+\langle1|)$,
$E_4=(|0\rangle+i|1\rangle)(\langle0|-i\langle1|)$,
$E_5=(|0\rangle+|2\rangle)(\langle0|+\langle2|)$,
$E_6=(|0\rangle+|1\rangle+|2\rangle)(\langle0|+\langle1|+\langle2|)$,
$E_7=(|0\rangle+|1\rangle+i|2\rangle)(\langle0|+\langle1|-i\langle2|)$.

Let an arbitrary unknown pure state in $\mathcal{H}_3$ be $|\phi\rangle=\sum_{k=0}^2a_ke^{i\theta_k}|k\rangle$.
Let $a_k$ be non-negative real numbers for $k=0,1,2$.
As $e^{i\pi}=-1$, we can modify the value of $\theta_k$ to guarantee $a_k\ge 0$.
Let $\theta_k$ be in the range of $[0,2\pi)$, as $e^{i\theta_k}=e^{i(\theta_k+2t\pi)}$ for integer $t$.
For the freedom choice of global phase, we let $\theta_0=0$.
The outcome probabilities can be calculated as follows:

$\mbox{tr}(E_k|\phi\rangle\langle\phi|)$=$a_k^2$, for $k=0,1,2$.\\
$\mbox{tr}(E_3|\phi\rangle\langle\phi|)$=$a_0^2+a_1^2+2a_0a_1\cos\theta_1$,\\
$\mbox{tr}(E_4|\phi\rangle\langle\phi|)$=$a_0^2+a_1^2+2a_0a_1\sin\theta_1$,\\
$\mbox{tr}(E_5|\phi\rangle\langle\phi|)$=$a_0^2+a_2^2+2a_0a_2\cos\theta_2$,\\
$\mbox{tr}(E_6|\phi\rangle\langle\phi|)$=$a_0^2+a_1^2+a_2^2+2a_0a_1\cos\theta_1+2a_0a_2\cos\theta_2+2a_1a_2\cos\theta_1\cos\theta_2+2a_1a_2\sin\theta_1\sin\theta_2$,\\
$\mbox{tr}(E_7|\phi\rangle\langle\phi|)$=$a_0^2+a_1^2+a_2^2+2a_0a_1\cos\theta_1+2a_0a_2\sin\theta_2+2a_1a_2\cos\theta_1\sin\theta_2-2a_1a_2\sin\theta_1\cos\theta_2$.\\

The coefficients of $a_k$ can be calculated by $E_k$, where $k=0,1,2$.
As the coefficient $a_k$ is non-negative, we have $a_k=\sqrt{\rm{tr}(E_k|\phi\rangle\langle\phi|)}$.
The remaining task is to determine $\theta_k$.

When only one element in $\{a_0,a_1,a_2\}$ is nonzero, it is the trivial case.
The state can be $|0\rangle$, $|1\rangle$ or $|2\rangle$.

When two elements in $\{a_0,a_1,a_2\}$ are nonzero, the state can also be determined.
For example, $a_0=0$ and $a_1,a_2\ne 0$.
We can write the state as $|\phi\rangle=a_1|1\rangle+a_2e^{i\theta_2}|2\rangle$.
The global phase of $\theta_1$ is extracted.
So $\theta_1=0$.
The remaining unknown coefficient $\theta_2$ can be calculated by the effect of $E_6$ and $E_7$.
If $a_1=0$ and $a_0,a_2\ne 0$, the state is $|\phi\rangle=a_0|0\rangle+a_2e^{i\theta_2}|2\rangle$.
The coefficient $\theta_2$ can be calculated by the effect of $E_5$ and $E_7$.
If $a_2=0$ and $a_0,a_1\ne 0$, coefficient $\theta_1$ can be calculated by the effect of $E_3$ and $E_4$.

When all elements in $\{a_0,a_1,a_2\}$ are nonzero, we let $\theta_0=0$.
Now we determine the remaining coefficients $\theta_1$ and $\theta_2$.
From the effect of $E_3$ and $E_4$, $\cos\theta_1$ and $\sin\theta_1$ can be calculated correspondingly, thus the coefficient $\theta_1$ can be uniquely determined.
After we know the values of $a_k$ and $\theta_1$, we can calculate $\cos\theta_2$ by the effect of $E_5$.
At the same time, $\sin\theta_2$ can be calculated by the effect of $E_6$ or $E_7$,
as $\cos\theta_1$ and $\sin\theta_1$ can not be both zero.
Then $\theta_2$ is uniquely determined.

Thus any pure state in $\mathcal{H}_3$ can be uniquely determined by the eight rank-1 positive self-adjoint operators.
These operators satisfies the condition in Theorem 1.
After the rank-1 conversion, we will get a PSIR-complete POVM with eight elements.
By a reference to Heinosaari \emph{et al.} \cite{Heinosaari_2013}, such POVM is one of the minimal possible resource.

\subsection{d=4}

For dimension $d=4$, the known result is that $\textbf{m}_0(4)=10$ \cite{Heinosaari_2013}.
There are five mutually unbiased bases.
Thus we can construct many rank-1 IC-POVMs with 16 elements.
Four orthonormal bases can distinguish all pure states in $\mathcal{H}_4$ \cite{Carmeli_2015}.
By theorem 1 and theorem 2, we can construct many PSIR-complete POVMs with 13 elements.
So the true value of $\textbf{m}_1(4)$ is in the range of $\{10,11,12,13\}$.

It is still not clear whether three bases can distinguish all pure states in $\mathcal{H}_4$.
No results show three bases would fail and no results give the potential support.
There are some partial answers to this question.
Three orthonormal bases consisting solely of product vectors are not enough.
In fact, even four product bases are not enough \cite{Carmeli_2015}.
Eleven is the minimum number of Pauli operators needed to uniquely determine any two-qubit pure state \cite{Ma_2016}.

We can conclude that there is no gap between $\textbf{m}_0(d)$ and $\textbf{m}_1(d)$ when $d=2,3$.
If a gap exists when $d=4$, three orthonormal bases are not enough to distinguish all pure states.
Consider the contrapositive form.
If three orthonormal bases can distinguish all pure states in $\mathcal{H}_4$, we can construct a PSIR-complete POVM with 10 elements by Theorem 2.

\section{ Adaptive $d+2k-2$ rank-1 operators for any dimensions}

Goyeneche \emph{et al.} took an adaptive method to demonstrate that any input pure state in $\mathcal{H}_d$ is unambiguously reconstructed by measuring five observables, i.e., projective measurements onto the states of five orthonormal bases \cite{Goyeneche_2015}. Thus $\sim5d$ rank-1 operators are needed.
The adaptive method is that the choice of some measurements is dependent on the result of former ones.
The fixed measurement basis is the standard, $\mathcal{B}_0=\{|0\rangle,\cdots,|d-1\rangle\}$.
We measure the pure state with this basis first.
The results of this basis will determine a subset $s_{d,1}\subset S_{d,1}$, where the input pure state belongs to.
They construct four bases $\{\mathcal{B}_1,\mathcal{B}_2,\mathcal{B}_3,\mathcal{B}_4\}$ to determine all pure states in $s_{d,1}$.

Let an arbitrary unknown input pure state in $\mathcal{H}_d$ be $|\phi\rangle=\sum_{s=0}^{d-1}a_se^{i\theta_s}|s\rangle$,  where $a_s$ is a non-negative real number and $\theta_s\in[0,2\pi)$ for $s=0,\cdots,d-1$.
We can extract the global phase to let one phase $\theta_s$ be 0.

Now we construct $d+2k-2$ adaptive rank-1 positive self-adjoint operators to determine this pure state, where $1\le k\le d$.
Thus at most $3d-2$ rank-1 elements are enough by adaptive strategy.

The first $d$ operators to be measured with are
\begin{equation}
E_s=|s\rangle\langle s|, s=0,\cdots,d-1.
\end{equation}
We can calculate the amplitudes $a_s$ by the effect of $E_s$,
$a_s=\sqrt{\mbox{tr}(E_s|\phi\rangle\langle\phi|)}$.
Then we keep track of the sites $\{s\}$ of nonzero amplitudes $\{a_s\}$ to determine a subset $s_{d,1}$.
Let $k$ be the number of nonzero amplitudes, $1\le k\le d$.

For example, the sites of nonzero amplitudes are $\{0,\cdots,d-1\}$.
Then $k=d$.
The subset $s_{d,1}$ is $\{\sum_{k=0}^{d-1}a_ke^{i\theta_k}|k\rangle: a_k\ne 0\}$.
The remaining $2d-2$ rank-1 operators are follows:
\begin{equation}
F_s=(|0\rangle+|s\rangle)(\langle0|+\langle s|), G_s=(|0\rangle+i|s\rangle)(\langle0|-i\langle s|);
\end{equation}
where $s=1,\cdots,d-1$.

We extract the global phase to make $\theta_0=0$.
We have the equations:
\begin{equation}
\left\{
\begin{aligned}
\mbox{tr}(F_s|\phi\rangle\langle\phi|) & = & a_0^2+a_s^2+2a_0a_s\cos\theta_s, \\
\mbox{tr}(G_s|\phi\rangle\langle\phi|) & = & a_0^2+a_s^2+2a_0a_s\sin\theta_s .
\end{aligned}
\right.
\end{equation}
%

From the assumption and measurement results of $E_s$, all the amplitudes $a_s$ are nonzero and known.
Then $\cos\theta_s$ and $\sin\theta_s$ can be calculated by the effect of $F_s$ and $G_s$.
All the coefficients $\theta_s$ can be uniquely determined.
Thus all coefficients of the unknown pure state in $\mathcal{H}_d$ are calculated.

The operators $E_s$ and $G_s$ appear in the construction of PSI-complete POVM given by Finkelstein \cite{Finkelstein_2004}. 
And operators $E_s$, $F_s/2$ and $G_s/2$ are some part of $d^2$ rank-1 elements in the IC-POVM constructed by Caves \emph{et al.} \cite{Caves_2002}.  
In fact, $F_s$ and $G_s$ can be the other types to calculate $\theta_s$.
For example, $F_s=(|1\rangle+|s\rangle)(\langle1|+\langle s|)$, $G_s=(|1\rangle+i|s\rangle)(\langle1|-i\langle s|)$, $s=0,2,\cdots,d-1$.

Now consider the general case, the sites of nonzero amplitudes are $\{n_0,\cdots,n_{k-1}\}$.
The subset $s_{d,1}$ is $\{\sum_{j=0}^{k-1}a_{n_j}e^{i\theta_{n_j}}|n_j\rangle:a_{n_j}\ne 0\}$.
The remaining $2k-2$ projections are as follows:
\begin{equation}
F_s=(|n_0\rangle+|n_s\rangle)(\langle n_0|+\langle n_s|), G_s=(|n_0\rangle+i|n_s\rangle)(\langle n_0|-i\langle n_s|);
\end{equation}
where $s=1,\cdots,d-k-1$.
Let the phase $\theta_{n_0}=0$.
With similar analysis, we can uniquely calculate $\cos\theta_j$ and $\sin\theta_j$ by the effect of $F_j$ and $G_j$.
All the phases $\theta_s$ and amplitudes $a_s$ of $|\phi\rangle$ can be uniquely determined.

\section{Conclusions}

We analyse the minimal number of elements in rank-1 PSIR-complete POVM.
The bound is in $[4d-3-c(d)\alpha(d),4d-3]$.
The lower bound of $3d-2$ is not tight except for $d=2$, 4.
For $d=2$, we construct many rank-1 POVMs with four elements which can distinguish all quantum states.
For $d=3$, we show that eight is the minimal number in a PSIR-complete POVM by construction.
For $d=4$, if $\textbf{m}_1(4)>10$, we can give a answer to an unsolved problem.
Three orthonormal bases can not distinguish all pure states in $\mathcal{H}_4$.
Finally, we construct $d+2k-2$ adaptive rank-1 positive self-adjoint operators to determine any input states in $\mathcal{H}_d$,  where $1\le k \le d$.
Thus we can determine an arbitrary unknown pure state in $\mathcal{H}_d$ with at most $3d-2$ rank-1 operators by adaptive strategy.

\section*{Acknowledgements}
This work was partially supported by National Key Research and Development Program of China under grant 2016YFB1000902, National Research Foundation of China (Grant No.61472412), and Program for Creative Research Group of National Natural Science Foundation of China (Grant No. 61621003).

\end{document}